\documentclass[11pt]{article}
\usepackage[utf8]{inputenc}
\usepackage[english]{babel}
\usepackage[T1]{fontenc}
\usepackage[margin=1.15in]{geometry}

\usepackage{amsmath}
\usepackage{amsthm}
\usepackage{amsfonts}
\usepackage{graphicx}
\usepackage{float}
\usepackage{pgfplots}
\pgfplotsset{compat=1.14, set layers}
\usepackage{bm}
\usepackage{mathtools}
\usepackage{subfig}
\usepackage{enumitem}

\usepackage{lipsum}
\usepackage{setspace}
\usepackage{caption}
\usepackage{accents}
\usepackage{mathtools}
\usepackage{todonotes}
\usepackage{xspace}
\usepackage[bottom]{footmisc}
\usepackage{framed}
\usepackage{thm-restate}
\usepackage{algorithm}
\usepackage[noend]{algpseudocode}
\usepackage[bottom]{footmisc}

\definecolor{refcolor}{rgb}{0.23, 0.27, 0.29}
\usepackage[breaklinks,colorlinks,urlcolor={refcolor},citecolor={refcolor},linkcolor={refcolor}]{hyperref}
\usepackage[capitalize, nameinlink]{cleveref}

\newtheorem{theorem}{Theorem}[section]
\newtheorem{lemma}[theorem]{Lemma}

\newtheorem{corollary}[theorem]{Corollary}
\newtheorem{definition}[theorem]{Definition}

\newtheorem{claim}[theorem]{Claim}
\newtheorem{observation}[theorem]{Observation}

\newcommand{\mpc}{\textsf{MPC}\xspace}
\newcommand{\ampc}{\textsf{AMPC}\xspace}

\newcommand{\pram}{\textsf{PRAM}\xspace}
\newcommand{\id}{\textsf{ID}\xspace}
\newcommand{\dht}{\textsf{DHT}\xspace}

\newcommand{\shrink}{\textsc{Shrink}\xspace}
\newcommand{\gshrink}{\textsc{ShrinkGeneral}\xspace}
\newcommand{\sar}{\textsc{ShrinkRecurse}\xspace}
\newcommand{\compose}{\textsc{Compose}\xspace}

\newcommand{\contract}{\textsc{Contract}\xspace}
\newcommand{\cc}{\textsc{ConnectedComponents}\xspace}

\newcommand{\whp}{w.h.p\xspace}

\newcommand{\larr}{\xleftarrow{}}

\newcommand{\logstar}[1]{\log^{*} #1}
\newcommand{\itlog}[2]{\log^{(#1)} #2}
\newcommand{\eps}{\varepsilon}

\newcommand{\ex}{\mathbb{E}}
\newcommand{\pr}{\text{Pr}}
\newcommand{\tspc}{T}
\newcommand{\mspc}{S}
\newcommand{\T}{B}

\Crefname{paragraph}{Paragraph}{Paragraphs}
\allowdisplaybreaks

\begin{document}

\begin{center}
	{\huge \bf Adaptive Massively Parallel Connectivity in Optimal Space} \\ \vspace{1cm}

	\begin{minipage}[H]{14.5cm}
		{\large \textbf{Rustam Latypov\footnotemark}, Aalto University -- \href{mailto:rustam.latypov@aalto.fi}{\texttt{rustam.latypov@aalto.fi}}} \vspace{0.5mm}\\
		{\large \textbf{Jakub Łącki}, Google Research, New York -- \href{mailto:jlacki@google.com}{\texttt{jlacki@google.com}}} \vspace{0.5mm}\\
		{\large \textbf{Yannic Maus}, TU Graz -- \href{mailto:yannic.maus@ist.tugraz.at}{\texttt{yannic.maus@ist.tugraz.at}}} \vspace{1mm}\\
		{\large \textbf{Jara Uitto}, Aalto University -- \href{mailto:jara.uitto@aalto.fi}{\texttt{jara.uitto@aalto.fi}}} \vspace{0.5mm}\\
	\end{minipage}

	\vspace{5mm}
	\begin{minipage}[H]{13.3cm}
		\begin{center}
			{\bf Abstract} \\ 
		\end{center}

		We study the problem of finding connected components in the Adaptive Massively Parallel Computation (\ampc) model. We show that when we require the total space to be linear in the size of the input graph the problem can be solved in $O(\logstar{n})$ rounds in forests (with high probability) and $2^{O(\logstar{n})}$ expected rounds in general graphs. This improves upon an existing $O(\log \log_{m/n} n)$ round algorithm. \\
		
		For the case when the desired number of rounds is constant we show that both problems can be solved using $\Theta(m + n \itlog{k} n)$ total space in expectation (in each round), where $k$ is an arbitrarily large constant and $\itlog{k}{}$ is the $k$-th iterate of the $\log_2$ function. This improves upon existing algorithms requiring $\Omega(m + n \log n)$ total space.

	\end{minipage}
\end{center}

\vfill
\thispagestyle{empty}
\footnotetext{Supported by the Academy of Finland, Grant 334238}

\newpage
\thispagestyle{empty}
\tableofcontents

\newpage
\pagenumbering{arabic}

\section{Introduction}
The Adaptive Massively Parallel Computation (\ampc) model is a computation model that captures the capabilities and challenges of modern platforms for processing massive data~\cite{ampc-constant, MPC-via-remote-access}.
In the \ampc model we have $M$ machines that communicate with each other, in synchronous communication rounds, each equipped with local space of size $\mspc$.
The machines communicate using a shared distributed hash table (\dht) (a distributed key-value store).
Within each round there is a read-only \dht containing the input to the round, and a write-only \dht where the machines write the output of the round.
Once the round completes, a new round begins and the output \dht from the previous round becomes the read-only input \dht for the next round.

The model has three challenging restrictions.
First, the space available to each machine, \mspc, is strictly sublinear in the input size, $N$.
Second, each machine can only read and write data of size at most $S$ within each round.
Third, the total space of all machines should be barely big enough to store all the input, that is $\tspc = \mspc \cdot M = O(N)$.

The \ampc model is an extension of the widely studied \mpc model.
The models differ in how the machines are allowed to communicate.
Specifically, in the \mpc model instead of writing data to \dht, within a round each machine can send messages to other machines, which are delivered  in the beginning of the following round.
The only restriction is that the total size of all messages sent to all machines in a round is at most $S$.
That is, the difference between the models is that in the \mpc model each machine in each round is \emph{given} a chunk of data to process, i.e., the messages it receives, while in the \ampc model each machine can \emph{choose} what data to read from the \dht. 
In particular the machine may use any value read within a round to adaptively decide what to read next (within the same round).

The \ampc model is particularly well suited to studying graph algorithms, and indeed several algorithmic problems have been solved more efficiently in \ampc compared to the \mpc model, including connected components~\cite{MPC-via-remote-access, ampc-constant}, maximal matching and independent set~\cite{ampc-constant, behnezhad2022time, hajiaghayi-contraction}, and minimum cut~\cite{DBLP:conf/spaa/HajiaghayiKOS22}.
All of these results are obtained in the regime when the available space per machine is sublinear in the number of vertices of the input graph, that is $\mspc = n^{\delta}$ for a constant $0 < \delta < 1$ for an input graph with $n$ vertices.
This regime is the most challenging (and the most desirable) regime for studying graph algorithms in the \mpc model.
At the same time, some of the fundamental unconditional lower bounds carry over from \mpc to \ampc~\cite{Charikar2020, roughgarden2018shuffles}.



Recent work on both \ampc and \mpc algorithms focused primarily on optimizing the number of rounds, which was motivated by the fact that synchronization in distributed systems is often an expensive step~\cite{DG08, suri2011counting}. 
However, if we consider the motivation behind the  the \ampc and \mpc models, the total space usage should also be highly correlated with empirical performance.
This is because in a vast majority of \ampc and \mpc algorithms, the the total space usage is determined by the maximum amount of communication that happens in any round.
In fact, in the usual case when each machine uses space that is linear in its input and output size, the total space usage and total communication are equal, up to constant factors.
As a result, the total space usage and the total communication can often be considered very good measures of how expensive a single round is.

This motivates the following question: what is the best round complexity that we can achieve if we require the space usage to be \emph{optimal}, that is, linear in the input size?
We address this question for the fundamental problem of finding connected components and give improved \ampc algorithms, which use optimal space.  

\subsection{Our Contributions}
We give improved algorithms for finding connected components in the \ampc model with sublinear space per machine, i.e., $S = n^\delta$ for any constant $\delta \in (0,1)$.
Our first result is an algorithm for finding connected components in forests.
Note that we say that an event holds with high probability (w.h.p.) if it holds with probability at least $1-1/n^{c}$, for a constant $c$ we can choose.

\begin{restatable}{theorem}{thmmainforest}
\label{thm:main-forest}
There exists a randomized $O(\logstar n)$-round \ampc algorithm that w.h.p.\ computes the connected components of an $n$-vertex forest and uses optimal total space.

More generally, there exists a randomized $O(k)$-round \ampc algorithm that w.h.p.\ computes the connected components of an $n$-vertex forest and uses  $O(n\log^{(k)}n)$ total space, for any $1 \leq k \leq \logstar{n}$.
\end{restatable}

This algorithm directly improves upon an existing result using $O(1)$ rounds and $O(n \log n)$ total space~\cite{MPC-via-remote-access}. 
We note that forest connectivity was used as a subroutine in some other \ampc algorithms~\cite{hajiaghayi-contraction, ampc-constant}.
We also give a new algorithm for the case of general graphs.

\begin{restatable}{theorem}{thmmaingeneral}\label{thm:maingeneral}
There exists a randomized \ampc algorithm for computing connected components of an $n$-vertex and $m$-edge graph $G$.
The algorithm runs in $2^{O(k)}$ rounds, each using $O(m + n \itlog{k}{n})$ total space in expectation, for any $1 \leq k \leq \logstar{n}$.
\end{restatable}

By setting $k = \logstar{n}$ in \Cref{thm:maingeneral} we obtain an algorithm using $2^{O(\logstar{n})}$ rounds and optimal space in expectation (in each round).
We note that for any constant $k$, $2^{O(\logstar{n})} = o(\itlog{k}{n})$.

This result improves upon two existing \ampc algorithms.
The first one is a $O(\log \log_{T/n} n)$-round algorithm which uses $O(T)$ total space~\cite{MPC-via-remote-access}.
The second one is a $O(1)$-round algorithm which uses either $O(m \log n)$ or $O(m + \log^2 n)$ space~\cite{ampc-constant}.
We note that our algorithm improves the space/round complexity tradeoff compared to both algorithms.

\subsection{Theoretical Motivation \& Related Work}

The problem of finding the connected components in an undirected graph is one of the central graph problems with many practical applications~\cite{Sahu-survey2020}. To understand the role of total space in the algorithm design in \ampc, let us first discuss the state-of-the-art of the connected components problem in \mpc.
In the \mpc model the problem can be solved in $O(\log D + \log \log_{T/n} n)$ rounds using total space of $T$, when the diameter of the input graph is at most $D$~\cite{Andoni2018, Behnezhad2019-near, Coy-det2022}.
Note that this running time becomes $O(\log D)$ when the total space is \emph{polynomially} larger than the input size.
Under the widely believed 1-vs-2 cycles conjecture~\cite{circuitLB18}, the $\Omega(\log D)$ round complexity is also the best one can hope for. In \ampc, this conditional hardness does not hold. 
The \dht in \ampc alleviates the dependency on $D$ and hence, given enough total space, the runtime collapses to $O(1)$.
Therefore in \ampc, we focus solely on the interplay between total space and the dependence on $n$ in the runtime, and significantly improve the existing tradeoff.

The fact that additional total space makes \mpc and \ampc algorithmic questions significantly easier seems to be a recurring theme for multiple problems in both models.
As an example, the commonly used exponentiation and round compression techniques in \mpc inherently require super-linear total space~\cite{Lenzen2010, czumaj2018round, ghaffari2019sparsifying}. These techniques are frequently used and even when combined with sophisticated additional techniques one still requires $\omega(n)$ total space, for example, see \cite{Dplus1coloring21}. 
Furthermore, somewhat surprisingly, it was recently shown that by increasing the available space significantly, yet still by a polynomial factor, one can essentially derandomize any \mpc algorithm~\cite{CDP21a}.

Algorithms with optimal memory have recently received attention in the \mpc and \ampc models.
In a very recent (and rather involved) result,  it was shown that connected components in forests can be computed in $O(\log D)$ rounds using optimal space~\cite{BLMOU23}.
This algorithm meets the conditional $\Omega(\log D)$ lower bound.
It improves over the $O(\log D + \log \log_{T/n} n)$-round algorithm, which can achieve $O(\log D)$ running time only at the cost of using much larger total space $T = n^{1+\Omega(1)}$.
As another example, there are efficient algorithms with optimal space for \emph{local constraint satisfactions problems}---think of the vertex coloring problem or the maximal independent set problem---in the case of on constant degree forests~\cite{BBFLMOU22}.


In the case of the \ampc model, it was shown that the problem of computing maximal matching can be solved in $O(\log \log n)$ rounds using optimal space, or $O(1)$ rounds using $\Omega(m + n^{1+\Omega(1)})$ space~\cite{ampc-constant}.
Very recently, the $O(1)$ runtime was also shown possible in the case of optimal space~\cite{behnezhad2022time}, only thanks to highly involved new ideas.

\subsection{Technical Challenges}

The common building block of our algorithms is the following insight.
Assume that given a graph with $n$ vertices and total space vertex $T$, in a single round we can reduce the problem to a problem on a graph of only $n / \exp(T/n)$ vertices.
We observe that by iterating this algorithm we can increase the amount of available space per vertex extremely quickly, even if initially the total available space is only linear in the number of vertices.
Our first technical contribution is showing that such a reduction is indeed possible both in the case of forests and general graphs. 

The reduction in the number of nodes is achieved by contracting sets of nodes of the graph.
The challenging part is symmetry breaking and ensuring that different vertices agree on what contractions should be performed.
Observe that on average the amount of communication per each vertex is exponentially smaller than the size of the contracted set it belongs to.

To illustrate some of the challenges involved, consider a path of length $l$.
A natural solution is to sample vertices of the path uniformly and contract each vertex of the path to the nearest sampled vertex (which can be done even in optimal total space).
If our goal is to shrink the size of this path by a factor of $2^B$ we could try sampling each vertex uniformly with probability $1/2^B$.
However, when $l = \Theta(2^B)$, this would imply that with constant probability no vertex along the path is sampled, and so the expected length of the path after the shrinking step is still $\Omega(l)$.
As a result, we need to use more involved sampling schemes.
In fact, we use two different methods for the two cases we consider.
We propose a new shrinking algorithm for forests and improve the space usage of an existing algorithm for general graphs.

Let us now explain how the shrinking procedure is helpful in obtaining low round complexity.
Consider the case of forests and assume that the available space per vertex is $B$.
In one iteration, we can shrink the number of vertices by a factor of roughly $2^B$, which increases the available space per vertex to $2^B$.
This additional space budget allows us to run the shrinking procedure with much larger sampling rate and increase the available space per vertex to $2^{2^B}$.
By continuing this process in only $O(\logstar{n})$ rounds we reach a state where the available space per each vertex is polynomial, in which case existing algorithms can solve the connected components problem in $O(1)$ rounds and optimal space.

\section{Preliminaries}\label{sec:preliminaries}
For any $x \in \mathbb{R}$ we define $\overline{\log}\; x$ as follows. For $x \geq 1$, $\overline{\log}\; x = \log_2 x$, and for $x < 1, \overline{\log}\; x = 1$.
For any integer $k \geq 0$, by $\itlog{k}$ we denote the $k$-th iterate of the $\overline{\log}$ function.
That is, $\itlog{0}{n} = n$, and for $k > 0$, $\itlog{k}{n} = \overline{\log}\; \itlog{k-1}{n}$.
Moreover, we define $\logstar{n}$ to be the minimum $k \geq 0$, such and $\itlog{k}{n} \leq 1$.
We also define the "inverse" of $\logstar{n}$, denoted by $2 \uparrow \uparrow k$.
We have $2 \uparrow \uparrow 0 = 1$, and for any integer $k > 0$, $2 \uparrow \uparrow k = 2^{2 \uparrow \uparrow (k-1)}$.

For a graph $G$ we say that a \emph{connected components labeling} (or CC-labeling for short) of $G$ is a mapping $M : V(G) \rightarrow A$ (where $A$ is an arbitrary set), such that for any two $u, w \in V(G)$ we have $M(u) = M(v)$ if and only if $u$ and $v$ belong to the same connected component of $G$.

\begin{definition}\label{def:cc-shrinking}
We say that an algorithm is \emph{connected-component shrinking} (or CC-shrinking for short) if it takes as input a graph $G$ and outputs a graph $H$ and a mapping $M$, such that given a CC-labeling of $H$ and the mapping $M$, one can compute a CC-labeling of $G$ in $O(1)$ \ampc rounds using optimal space.  
\end{definition}

A CC-shrinking algorithm essentially reduces the problem of finding connected components in $G$ to solving the problem on $H$. For any CC-shrinking algorithm, we generally refer to the $O(1)$-round operation that produces the CC-labeling of $G$ from the CC-labeling of $H$ as $\compose(H,G)$.  So, the operation \compose can be seen as the inverse operation of a CC-shrinking algorithm. 


In our algorithms we use multiple CC-shrinking algorithms. One of them is a standard vertex contraction algorithm, which we denote by $\contract(G, C)$.
It takes a graph $G$ and a mapping $C : V(G) \rightarrow A$, where $A$ is an arbitrary set, and contracts (merges together) groups of vertices that are assigned the same value by $C$.
Any resulting parallel edges are merged into one and  self-loops are removed.
This is a commonly used subroutine, which can be implemented in $O(1)$ (\textsf{A})\mpc rounds using optimal space~\cite{Behnezhad2019-near}.

\begin{observation}\label{obs:contract}
$\contract(G,\cdot)$ is a CC-shrinking algorithm. 
\end{observation}

For simplicity, we assume that \contract only returns a graph, and not the mapping mentioned in \Cref{def:cc-shrinking}, as the mapping that can be used to recover the CC-labeling of its input is actually one of its parameters.
Another essential tool used throughout the paper is the following concentration bound.

\begin{lemma}[Hoeffding's concentration bound]
\label{lem:hoeffding} 
Let $X_1,\ldots, X_{\ell}$ be independent random variables with $X_i\in [a_i,b_i]$. Let $X=\sum_{i=1}^{\ell}X_i$ and let $t>0$. 
Then the following holds
\begin{align*}
\Pr(X-\ex(X)\geq t)\leq \exp\left(-\frac{2t}{\sum_{i=1}^{\ell}(b_i-a_i)^2}\right)~.
\end{align*}
\end{lemma}
\section{Forest Connectivity}
\label{sec:ForestConnectivity}
In this section, we present an algorithm for solving forest connectivity in $O(\log^* n)$ rounds using optimal space. The forest connectivity problem is the undirected graph connectivity problem when the input graph is a forest. More formally, we prove the following theorem. 
\thmmainforest*


\newcommand{\slc}{\textsc{ShrinkLargeCycles}\xspace}
\newcommand{\ssc}{\textsc{ShrinkSmallCycles}\xspace}
\newcommand{\sscf}{\textsc{Standard-Cycle-CC}\xspace}

\begin{algorithm}[H]
\caption{Finding connected components in a forest.}\label{alg:forest}
\begin{algorithmic}[1]

\Function{ConnectedComponentsForest}{$G$}
    \State{Reduce to cycle-connectivity (\Cref{obs: foresttocycle}).} \label{l:first}
    \State{$G' \larr \slc(G)$} \label{l:second}

    \State{$\T \larr 100$} 
    
    \While{$|V(G')|> n/\log n$} \label{l:ssc1}

        \State{$G' \larr \ssc(G',\T)$} \label{l:ssc2}

        \State{$\T \larr \min\{2^{\T},(\eps\log n)/100\}$ (every second iteration)} \label{l:ssc3}
        
    \EndWhile

    \State{\Return{$\sscf(G')$}}\label{l:end}

\EndFunction

\end{algorithmic}
\end{algorithm}

\textbf{High level view on the algorithm:} See \Cref{alg:forest} for the pseudocode of the algorithm. Throughout this section let $\eps=\delta/10$; recall, that the local space of a machine is $\mspc=n^{\delta}$. 
The algorithm is a sequence of several CC-shrinking algorithms, which conceptually produce a sequence of graphs $G'_1, \ldots, G'_r$, i.e., $G'_{i+1}$ is the output of a CC-shrinking algorithm running on $G'_i$.
At the very end (\Cref{l:end}) it computes the connected components of the graph $G'_r$.
Since the sequence of graphs is obtained by running a CC-shrinking algorithm, we can now obtain connected components of the input graph by a proper sequence of \compose calls.
However, we skip these calls in the pseudocode for simplicity.
Note that since compose runs in $O(1)$ \ampc rounds, the running time of each compose call can be charged to the step which produces one of the graphs $G'_i$.

\Cref{alg:forest} starts with a couple of easy reductions.
The first step (\Cref{l:first}) is to reduce the forest-connectivity problem to the cycle-connectivity problem 
by transforming each forest into a cycle using an Eulerian tour. 
As observed by~\cite{MPC-via-remote-access} this reduction can be done in $O(1)$ \ampc rounds by 
directly implementing the classic \pram construction \cite{forest-to-cycles}.


Let us describe the high-level idea behind the construction.
Let us first replace each edge with two oppositely directed edges.
Consider a vertex $v$ of degree $d$ (in the undirected graph).
Denote its incident edges as $\overleftarrow{e_0}, \overrightarrow{e_0}, \overleftarrow{e_1}, \overrightarrow{e_1}, \ldots, \overleftarrow{e_{d-1}}, \overrightarrow{e_{d-1}}$, where $\overleftarrow{e_i}$ and $\overrightarrow{e_i}$ is a pair of edges to and from the same neighbor of $v$.
We then replace $v$ with $d$ vertices $v^0, \ldots, v^{d-1}$, where each $v^i$ has two incident edges $\overleftarrow{e_i}$ and $\overrightarrow{e_{(i+1) \bmod d}}$.
The algorithm applies this vertex splitting to all vertices in parallel, and then makes each edge undirected. 
This maps a tree containing $k > 1$ vertices to a cycle of length $2k-2$~\cite{forest-to-cycles}.

\begin{observation}[Forests to Cycles]\label{obs: foresttocycle}
There is a deterministic CC-shrinking algorithm, which takes a forest on $n$ vertices and outputs a collection of vertex-disjoint cycles on at most $2n$ vertices.
It can be implemented in $O(1)$ \ampc rounds using optimal space.
\end{observation}

The second step (\Cref{l:second}) is to ensure that each cycle has length at most $O(n^{\eps/2})$ by applying the following lemma.

\begin{lemma}[Corollary 8.1, \cite{MPC-via-remote-access}] \label{lem:ShrinkLargeCycles}
\label{lem:makeCyclesntoEps}
    There is a randomized CC-shrinking \ampc algorithm ($\slc$) which can be applied to a set of cycles to reduce the size of each individual cycle to $O(n^{\eps/2})$ w.h.p.\. It can be implemented in $O(1)$ \ampc rounds using optimal space.
\end{lemma}




\begin{figure*}
\begin{framed}

\hspace{6mm}
\begin{minipage}{14.5cm}
    \begin{enumerate}
    \item[Ranks.~~] Define distribution $\pi_{\T}$ such that $\pi_{\T}(i)=C_{\T}/2^i$ for $i \in \{1,\dots,\T\}$ and $\pi_{\T}(i)=0$ otherwise, where $C_{\T}=1/(1-2^{-\T})$. For every vertex $v$ assigned to a machine, sample a rank $r(v)$ from $\pi_{\T}$.
    \item[Step~1.~~] For each vertex $v$ assigned to a machine and both directions of the cycle, traverse until (i) $v$ loops back to itself or (ii) $v$ encounters a vertex $u$ such that $r(u) \geq r(v)$. While traversing, vertex $v$ stamps every other vertex it encounters with its rank $r(v)$.
    In the former case, $v$ contracts$^*$ the whole cycle. In the latter case, for all highest rank nodes $v$ and $u$ connected via a segment of strictly lower rank nodes, w.l.o.g. $\id(v) > \id(u)$ and $v$ contracts$^*$ the segment between $v$ and $u$ in the cycle.
    \item[Step~2.~~] For each vertex $v$ assigned to a machine, traverse its $16\T$-hop neighborhood. If it contains the whole cycle and $v$ is its highest \id vertex, $v$ contracts$^*$ the whole cycle. If it does not contain the whole cycle and $v$ is the highest \id vertex in its $16\T$-hop neighborhood, $v$ contracts$^*$ its $4\T$-hop neighborhood.
\end{enumerate}
\hspace{10mm}$^*$apply \contract from Preliminaries (\Cref{sec:preliminaries})  
\end{minipage}
\end{framed}

\caption{$\ssc(G,\T)$}
\label{fig:ssc}
\end{figure*}

Now comes the most challenging part of our algorithm (\Cref{l:ssc1,l:ssc2,l:ssc3}), which further contracts the cycles such that the number of vertices remaining in these contracted cycles sees a significant drop compared to the overall global memory of $\Theta(n)$, while never exceeding the total space bound of $O(n)$.
Once we have reduced the number of vertices to $n'=O(n /\log n)$, we can finish the remaining instances with the algorithm $\sscf$ of  \cite{MPC-via-remote-access}, see \Cref{lem:ShrinkSmallCyclesFast}. 

Let us now detail on the main part of our algorithm (\Cref{l:ssc2}: \ssc, \Cref{fig:ssc}) that works in $O(\logstar n)$ iterations, and after the $i$-th iteration we guarantee that the number of remaining alive vertices has dropped to $n_i\leq n \cdot \frac{1}{2 \uparrow\uparrow i}$. For the first iteration, let $\T$ be a sufficiently large constant. Now, every vertex picks one out of $\T$ ranks according to a truncated geometric distribution. Ignoring the rescaling factor ensuring that we obtain a proper probability distribution, this means that a vertex picks rank $i$ with probability $1/2^i$. These ranks are chosen independently for all alive vertices. Then, each vertex probes the cycle around it in a single \ampc-round. The probing of $v$ stops in one direction if $v$ sees a vertex  of the same or higher rank. This results in every highest rank vertex knowing neighboring cycle segment(s) of vertices of lower rank. 
Breaking symmetry by \id{}s, the highest rank vertices can (collectively) contract all other vertices in the cycle. Hence, the number of vertices reduces to the number of vertices with the highest rank. The highest rank on a cycle can be any of the ranks $1,\ldots,\T$ and depends on the randomness of the vertices.

We show that overall the number of queries of this algorithm is $O(n'\cdot \T)$, if there are $n'$ vertices remaining in the graph at the start of the iterations (\Cref{lem:globalQueries}). Additionally, we show that two of these iterations w.h.p.\ reduce the number of remaining vertices in the graph from $n$ to $n/2^{\T}$ (\Cref{lem:whpProgress}). The main benefit of this reduction is that we obtain (on average) $2^{\T}$ words of memory per vertex that we can leverage in the next iteration. We do so, by increasing $\T$ to $2^{\T}$ every second iteration. 
Increasing $\T$ exponentially ensures that after $O(\logstar n)$ iterations, we have reduced the number of vertices to $O(n/\log n)$. 

The initial reduction of the maximum cycle length ensures that no vertex ever queries more than $n^{\eps}\leq n^{\delta}$ vertices in one iteration (one \ampc-round) of the algorithm. The most challenging part is to bound the query complexity in \Cref{lem:globalQueries} and the vertex drop in \Cref{lem:whpProgress}. For both of them, we first analyze the expectation of the respective term, which is then turned into a w.h.p.\ guarantee via an application of Hoeffding's concentration bound. We cannot obtain a w.h.p.\ guarantee on each individual cycle. But, as we have more than $n/\log n$ vertices left in the graph (recall, that otherwise we can use the algorithm of \Cref{lem:ShrinkSmallCyclesFast}) and each cycle is of length at most $n^{\eps}$, we have $\ell\geq \Omega(n^{1-\eps}/\log n)$ cycles left in the graph, providing the necessary handle for concentration. 

In particular, the vertex drop is challenging, as we can only bound the expected vertex drop on a cycle of $k$ vertices by $O(k/2^{\T}+2\T)$. Even, if we would meet this expectation on all cycles, the additive $2\T$ term would be insufficient for obtaining a global drop in the number of vertices by a factor $O(1/2^{\T})$. Intuitively, that's the case because the additive factor of $2\T$ has a significant (relative) impact for small cycles. Hence, each iteration is additionally equipped with a deterministic phase that  removes $\min\{k,8\T\}$ vertices on a cycle of size $k$. It is difficult to analyze the (expected) vertex reduction of that second deterministic phase, as it depends on whether the first phase reduced the number of vertices on a cycle to less than $8\T$ vertices or not (which happens according to some difficult to grasp probability distribution). Hence, in our analysis, we analyze both steps (the randomized rank-based one and the deterministic one) combined which shows the desired (expected and w.h.p.) drop in the number of vertices. 

We provide the statement from prior work that can solve the cycle-connectivity problem with an additional $\Theta(\log n)$-factor of global memory available (\Cref{l:end}).

\begin{lemma}[Theorem 5, \cite{MPC-via-remote-access}] \label{lem:ShrinkSmallCyclesFast}
    There is a randomized \ampc algorithm (\sscf) that solves the connectivity problem on a collection of disjoint cycles on $n$ vertices w.h.p.\. The algorithm runs in $O(1)$ \ampc rounds and uses $O(n \log n)$ total space.
\end{lemma}

The remainder of the section is dedicated to proving the most involved part of our algorithm, $\ssc$ in \Cref{fig:ssc}. We refer to one execution of \ssc as an \emph{iteration}. First we show that picking ranks in $\ssc$ actually follows a probability distribution.

\begin{claim} \label{claim:correctDistribution}
     $\pi_{\T}$ is a probability distribution on $\{0,1,\dots,\T\}$ for any integer $\T>0$.
\end{claim}

\begin{proof} We have $\pi_{\T}(i)\geq 0$ for any $i$, and 
    \begin{align*}
        \sum_{i=1}^{\T} C_{\T}/2^i & = C_{\T} \sum_{i=1}^{\T} 1/2^i = C_{\T} (1-2^{-\T}) \\
        & = 1/(1-2^{-\T}) \cdot (1-2^{-\T}) = 1 \qedhere
    \end{align*} 
\end{proof}

We continue with a claim that we need in order to obtain the bounds in our probabilistic analysis.

\begin{claim} \label{claim:queriesAnalysis}
\label{claim:easyBound}~

\begin{enumerate}[leftmargin=*]
\item For $0<x<1$ we have $\sum_{i=0}^{\infty}i\cdot x^{i-1}=\frac{1}{(1-x)^2}$.

\item For integers $0\leq x\leq y$ we have
$\sum_{j=x}^{y}j\cdot 2^{-j}=2\big(\big(\frac{1}{2}\big)^x-\big(\frac{1}{2}\big)^{y+1}\big)$~.
\end{enumerate}
\end{claim}

\begin{proof}~

\begin{enumerate}[leftmargin=*]
\item Let $f(x)=\sum_{i=0}^{\infty}x^i$. As $f(x)$ is a geometric sum and $0<x<1$, we obtain $f(x)=\frac{x}{1-x}$. Observe that $\sum_{i=0}^{\infty}i\cdot x^{i-1}=\frac{d}{dx} f(x)=\frac{1}{(1-x)^2}$, proving the claim. 

\item Re-ordering the sums\footnote{It is easiest to see this re-ordering when $x=0$ and $y=\infty$. Write the terms of the sum into a triangular grid with $1/2$ in the first row, $1/4, 1/4$ in the second row, $1/8,1/8,1/8$ in the third row. Now, if we sum the values in the first column we obtain $\sum_{i=1}^{\infty}2^{-j}$, for the second column we obtain half of that, and in general for the $i+1$ column we obtain half of the previous column.}, which does not change the limit as  the series is absolutely converging, and using the geometric sum we obtain
\begin{align*}
\sum_{j=x}^{y}\frac{j}{2^j}=\sum_{i=x}^{y}2^{-i}\cdot \sum_{j=1}^{\infty}2^{-j}=\sum_{i=x}^{y}\left(\frac{1}{2}\right)^i=2\left(\frac{1}{2^x}-\frac{1}{2^{y+1}}\right)~. & \qedhere
\end{align*}
\end{enumerate}

\end{proof}

\subsection{Query complexity}
We begin with bounding the expected number of queries per vertex. 
\begin{lemma}[Expected Queries per vertex]
\label{lem:expectedQuery}
    In Step~1 of an iteration, the number of queries made by a vertex until it hits a vertex with higher or equal rank is at most $4\T$ in expectation.
\end{lemma}

\begin{proof}
     Consider an arbitrary vertex $v$ in a cycle of length $k$ and let $X$ be the random variable describing the number of queries made by $v$ in one direction. The probability that we have to query $i$ vertices before finding a vertex with higher or equal rank is the probability that the $i$:th vertex has a higher or equal rank than $v$ and that all the $i-1$ vertices in between have a strictly lower rank than $v$. Let $p_j$ denote $C_\T/2^j$, which is the probability that a vertex draws rank $j$ from $\pi_\T$. If vertex $v$ has rank $j<\T$, the expected number of queries is at most
     \begin{align*}
         \ex(X\mid rank(v)=j)\leq  \sum_{i=1}^k i \cdot (1-p_j)^{i-1} \cdot 2p_j, 
     \end{align*}
     
     where $(1-p_j)^{i-1}$ is an upper bound on the probability that the $i-1$ vertices between $v$ and the $i$:th vertex have rank $<j$, and $2p_j$ is an upper bound on the probability that the $i$:th queried vertex has rank between $j$ and $\T$ (inclusive). If $v$ has rank $\T$, the expected number of queries is
     \begin{align*}
          \ex(X\mid rank(v)=\T)\leq \sum_{i=1}^k i \cdot (1-p_{\T})^{i-1} \cdot p_{\T}, 
     \end{align*}

     where $(1-p_{\T})^{i-1}$ is the probability that the $i-1$ vertices between $v$ and the $i$:th vertex have rank $<\T$, and $p_{\T}$ is the probability that the $i$:th queried vertex has rank exactly $\T$. By combining the aforementioned cases and applying the total law of expectation we obtain the following upper bound. 
     \begin{align*}
     \ex(X) & =\sum_{j=1}^{\T}\Pr(rank(v)=j)\cdot \ex(X \mid rank(v)=j) \\
         &\leq \left( \sum_{j=1}^{\T-1} p_j \sum_{i=1}^k i (1-p_j)^{i-1} \cdot 2 p_j \right) + \left( p_{\T} \sum_{i=1}^k i  (1-p_{\T})^{i-1} \cdot p_{\T} \right) \\
         &= \left( \sum_{j=1}^{\T-1} 2p^2_j \sum_{i=1}^k i (1-p_j)^{i-1} \right) + \left( p^2_{\T} \sum_{i=1}^k i (1-p_{\T})^{i-1} \right) \\
         &\stackrel{(*)}{\leq} \sum_{j=1}^{\T} 2p^2_j \sum_{i=1}^{\infty} i  (1-p_j)^{i-1} \\
         &\stackrel{(\dag)}{\leq}  \sum_{j=1}^{\T} \frac{2p^2_j}{p_j^2}  = 2\T~.
     \end{align*}

At $(*)$ we combine the two terms and sum to infinity instead of $k$, and at $(\dag)$ we apply \Cref{claim:queriesAnalysis} with $x = 1-p_j$. Vertex $v$ queries in both directions of the cycle, so the expected number of queries is $4\T$.
\end{proof}

By linearity of expectation, from \Cref{lem:expectedQuery} we deduce that the global expected query complexity is $O(n'\cdot \T)$ when we are left with $n'$ alive vertices at the beginning of the iteration. We use the Hoeffding's inequality to turn this expected guarantee into a w.h.p.\  bound on the global number of used queries. The quality of Hoeffding's bound depends on the range of the used random variables. The query complexity of a vertex is in $\{1,\ldots,n^{\eps}\}$, as each cycle is of length at most $n^{\eps}$. As the outcome of queries of vertices on the same cycle are not independent, we need to apply Hoeffding's inequality with one random variable measuring the number of queries on each cycle, which indeed are independent. Intuitively, the large number of cycles (recall that each cycle has length $\leq n^{\eps}$ and we have $\Omega(n/\log n)$ vertices remaining) provides the necessary concentration around the expected query complexity. 

\begin{lemma}[Global number of queries]
\label{lem:globalQueries}
Let $n'$ be the number of vertices at the beginning of one iteration. Then, w.h.p.\ the total number of queries used in the iteration by all vertices is at most $O(n'\cdot \T)$.
\end{lemma}
\begin{proof} We first focus on Step~1 of an iteration. 
Let $\ell\geq n^{10\eps}/n^{\eps}=n^{9\eps}$ be the number of cycles and let $k_1,\ldots,k_{\ell}$ be the number of vertices in these cycles before the current iteration. Let $S_1, \ldots,S_{\ell}$ be the random variable (depending on the randomness of the current iteration) that described the total number of queries performed by all vertices in the respective cycle (in the current iteration). 
Note that $S_{i}\in [k_i,\ldots,k_i^2 ]$ as none of the $k_i$ vertices of the cycle performs more than $k_i$ queries. Due to \Cref{lem:expectedQuery} and  linearity of expectation we have $\ex[S_i]\leq \T\cdot k_i$ for all $i\in[\ell]$ and the random variables $S_{i}\in [k_i,\ldots,k_i^2 ]$ are independent. 

Let $n'=\sum_{i=1}^{\ell}k_i$ and define  $S=\sum_{i=1}^{\ell} S_i$.  Let $\mu=n'\cdot 4\T$ and observe that $\ex(S)\leq \mu$ by linearity of expectation. 
We apply Hoeffding's inequality (\Cref{lem:hoeffding}) on the $\ell$ independent random variables and obtain that w.h.p.\ the total number of queries in the first step is bounded by $n'\cdot 8 \T$. More detailed, we obtain
\begin{align*}
    \Pr(S\geq 2\mu) & \leq \Pr(S-\ex(S)\geq \mu)
    \leq \exp\left(-\frac{2\mu^2}{\sum_{i=1}^{\ell}k_i^4}\right) \\
    &\stackrel{(k_i\leq n^{\eps})}{\leq} \exp\left(-\frac{2\mu^2}{\sum_{i=1}^{\ell} n^{4\eps}}\right) \leq \exp\left(-\frac{n^{20\eps}}{\ell n^{4\eps}}\right) \\
    &\leq \exp\left(n^{7\eps}\right)\leq 1/n^4~. \qedhere
\end{align*}

In the second step of an iteration each vertex queries at most $32\T$  queries per vertex, or $O(\T\cdot n')$ queries in total. Hence, the total query complexity over both steps combined is at most $O(\T\cdot n')$.
\end{proof}

\subsection{Measure of progress (vertex drop per iteration)}

We first prove that the second step of an iteration removes at least $\min\{8\T,k\}$ vertices from a cycle of length $k$, as this property will be used in the analysis of the total vertex drop per iteration (\Cref{lem:expectedvertexDrop}).

\begin{lemma} \label{lem:step2proof}
    Step 2 of an iteration removes at least $\min\{8\T,k\}$ vertices from a cycle of length $k$. 
\end{lemma}

\begin{proof}
    If $k \leq 32\T$, then the cycle is within the $16\T$-hop neighborhood of every vertex and the highest \id vertex will compress the whole cycle, effectively removing $\min\{8\T,k\}$ vertices. If $k>32\T$, then at very least the highest \id vertex of the cycle will compress its $4\T$-hop neighborhood ($8\T$ vertices). The compressions do not overlap due to the condition that a compressing vertex has to be the highest \id vertex in its $16\T$-hop neighborhood.
\end{proof}

Next, we analyze both steps of an iteration simultaneously. Recall that in the first step, the highest rank nodes in a cycle contract all other nodes such that contractions do not overlap due to the following claim.

\begin{claim}
    In Step 1 of an iteration, the highest rank nodes in a cycle contract all other nodes such that contractions do not overlap.
\end{claim}

\begin{proof}
    We have to prove that the highest rank nodes know that they are the highest rank nodes (so that they can perform contractions) and that no node gets contracted by two different highest rank nodes.
    
    The former holds due to every node $v$ stamping every node they visit with their rank $r(v)$. This implies that every node in the cycle will be stamped with the highest rank in the cycle. Knowing it, every node knows whether or not they are a highest rank node. The latter holds due to the symmetry breaking via \id{}s.
\end{proof}

\begin{lemma}
\label{lem:expectedvertexDrop}
    Consider a cycle with $k$ vertices at the beginning of one iteration. The expected number of vertices of the cycle after the iteration is bounded by $2k/2^{\T} + 1/2^{\T}$.
\end{lemma}

\begin{proof}
Consider a cycle with $k$ vertices at the beginning of the iteration. Let $\tau$ be the random variable describing the largest rank on the cycle and let $1\leq X\leq k$ be the random variable describing the number of vertices whose rank equals the largest rank $\tau$. We aim to find an expression of the expectation of $X$. For that purpose fix some $j\in \{1,\ldots,k\}$ and bound the probability that $X$ equals $j$. We consider the cases of $\tau<\T$ and $\tau=\T$ separately. 

We begin with bounding $\Pr(X=j\wedge \tau<\T)$. Fix some $i<\T$ and some set $M$ of $j$ vertices in the cycle. Set $p_i=1/2^i$. The probability that $\tau=i$ is the maximum rank appearing on the cycle and attained by all of these fixed $j$ vertices is at most $2 \cdot 2^{-j} p_i^j(1-p_i)^{k-j}$. Excluding the leading factor 2, the previous expression is exactly the scenario of $k$ players playing the coin tossing game of \Cref{claim:coinToss} for $\T$ rounds, and where $j$ players get the highest value $i$. The factor $2^{-j}$ appears because the $i+1$-th coin toss has to be false for those $j$ vertices, the factor $p_i^j$ appears because the first $i$ coin tosses have to be true for these $j$ vertices, and the factor $(1-p_i)^{k-j}$ appears because all other vertices should have one of the ranks $1,\ldots,i-1$ which happens with probability $1-p_i$ (independently) for each vertex. Due to the leading coefficient $(1/2)^{j \cdot B}$ in \Cref{l:proof} of \Cref{claim:coinToss}, we can upper bound the probability of $\Pr(X=j\wedge \tau<\T)$ using a factor 2 in the expression. 

There are $\binom{k}{j}$ different sets of size $j$. Hence, we obtain the following probability
\begin{align*}
    \pr \big(X=j \wedge \tau <\T\big) &= \sum_{i=1}^{\T-1} \Pr(X=j \wedge \tau=i) \\
    &= \sum_{i=1}^{\T-1} 2^{1-j} \binom{k}{j} p_i^j(1-p_i)^{k-j} \\
\end{align*}

The expectation of $X$ is the following. 
\begin{align*}
    \ex(X) &= \sum_{j=1}^kj\cdot \Pr(X=j) \\
    &= \sum_{j=1}^k j\cdot \Pr(X=j \wedge \tau< \T) +\sum_{j=1}^k j\cdot \Pr(X=j \wedge \tau= \T)~.
\end{align*}

Recall that after the randomized procedure (which contracts the cycle into the highest rank vertices in Step 1), there is a deterministic procedure (Step 2), which removes at least $\min\{8\T,k\}$ vertices from a cycle of length $k$ (\Cref{lem:step2proof}). Hence, the expected number of remaining vertices after \ssc is 
\begin{align*}
    &\leq \sum_{j=1}^{8\T} 0 \cdot \Pr(X=j \wedge \tau< \T) + \sum_{j=8\T+1}^k (j-8\T) \cdot \Pr(X=j \wedge \tau< \T) ~+ \\
    &~~~\sum_{j=1}^{8\T} 0 \cdot \Pr(X=j \wedge \tau= \T) + \sum_{j=8\T+1}^k (j-8\T) \cdot \Pr(X=j \wedge \tau= \T) \\
    &\leq \sum_{j=8\T+1}^k (j-8\T) \cdot \Pr(X=j \wedge \tau< \T) + \sum_{j=1}^k j\cdot \Pr(X=j \wedge \tau= \T)~.
\end{align*}

We bound the terms separately. The first term bounds by
\begin{align*}
    &\sum_{j=8\T+1}^k (j-8\T) \cdot \Pr(X=j \wedge \tau< \T) \\
    &\leq \sum_{j=8\T+1}^k j \cdot \Pr(X=j \wedge \tau< \T) \\
    &\leq \sum_{j=8\T+1}^k \sum_{i=1}^{\T-1} j \cdot 2^{1-j} \binom{k}{j} p_i^j(1-p_i)^{k-j} \\
    &\stackrel{(*)}{\leq} \sum_{j=8\T+1}^k \sum_{i=1}^{\T-1} j \cdot 2^{1-j}  
    \leq 2\T \sum_{j=8\T+1}^k j \cdot 2^{-j}  \\
    &\stackrel{(\dag)}{\leq} 2\T \cdot 2 ( (1/2)^{8\T+1} - (1/2)^{k+1}) \\
    &\leq 4\T \cdot (1/2)^{8\T+1}   
    \leq \frac{4\T}{2^{8\T+1}}
    < \frac{1}{2^{\T}}~.
\end{align*}

At $(*)$ we used that $\binom{k}{j}p_i^j(1-p_i)^{k-j}\leq 1$, and at $(\dag)$ we use \Cref{claim:easyBound} with $x=8\T+1$ and $y=k$. This holds as the value equals the probability of having $j$ successes appearing in $k$ Bernoulli trials with probability $p_i$. 

For the second term first let $1\leq Y\leq k$ be the random variable describing the vertices that pick rank $\T$. Note that $Y=j$ and $X=j\wedge \tau=\T$ are the same events. Let $S$ be the set of vertices of the cycle. We obtain the following. 
\begin{align*}
    \sum_{j=1}^k j\cdot \Pr(X=j \wedge \tau= \T) & =\sum_{j=1}^k j\cdot \Pr(Y=j)=\ex[Y] \\
    &= \sum_{v \in S} \Pr(rank(v)=\T) \\
    &= k \cdot p_{\T} = k \cdot C_{\T}/2^{\T} \leq 2k/2^{\T}.
\end{align*}

In total, we obtain that the expected number of remaining vertices is $2k/2^{\T} + 1/2^{\T}$.
\end{proof}

\begin{claim} \label{claim:coinToss}
Consider the geometric distribution $\pi$ such that $\pi(i)=1/2^i$ for $i\geq1$ and $\pi(i)=0$ otherwise. The probability of a player sampling $i$ from distribution $\pi$ is equivalent to the probability of obtaining value $i$ in the following coin tossing game. A player gets value 1 and starts tossing a fair coin repeatedly. Upon succeeding a flip, she increases her value by 1. Upon failing, the game ends. 

The analogy can be extended to a truncated geometric distribution $\pi_\T(i)$ such that $\pi_\T(i)=C_\T/2^i$ for $i \in \{1,\dots,\T\}$ and $\pi_\T(i)=0$ otherwise, where $C_\T=1/(1-2^{-\T})$ (\Cref{claim:correctDistribution} proves that $\pi_\T$ is a distribution). The probability of a player sampling $i \in \{1,\dots,\T\}$ from distribution $\pi_\T$ is equivalent to the probability of obtaining value $i$ in the following coin tossing game. A player initiates $q=1$ and then repeatedly tosses a fair coin. Upon succeeding a flip, she changes $q$ to $(q \hspace{-1mm}\mod \T) +1$. Upon failing, the game ends.
\end{claim}

\begin{proof}
    For $\pi$, the probability of a player obtaining value $i\geq1$ via the coin tossing game is $(1/2)^{i-1}\cdot(1-1/2) = 1/2^i$, which is equal to $\pi(i)$. 
    
    For $\pi_\T$, the probability of a player obtaining value $i \in \{1,\dots,\T\}$ via the coin tossing game is 
    \begin{align}
        \sum_{j=0}^\infty & (1/2)^{j \cdot \T} \cdot (1/2)^{i-1}\cdot(1-1/2) = \sum_{j=0}^\infty (1/2)^i \cdot (1/2^\T)^j \label{l:proof} \\
        &= \frac{(1/2)^i}{1-(1/2)^\T} \nonumber = C_\T/2^i, \nonumber
    \end{align}
    which is equal to $\pi_\T(i)$.
\end{proof}

\begin{lemma}
\label{lem:whpProgress}
Let $\T\leq (\eps \log n)/100$. Consider some iteration of the algorithm and let $n'$ be the number of vertices in connected components (cycles) with more than one vertex. 

If $n'\geq n^{10\eps}$, then w.h.p.\ the number of vertices in connected components with more than one vertex at the end of the iteration is at most $6n'/2^{\T}$.
\end{lemma}
\begin{proof}
Fix one iteration of the algorithm. Let $\ell$ be the number of remaining connected components (cycles) at the beginning of the iteration. 
Due to \Cref{lem:makeCyclesntoEps} each cycle is of length at most $n^{\eps}$.

Let $k_1,\ldots,k_{\ell}$ be the number of vertices in these cycles before the current iteration. Let $\bar{k}_1, \ldots,\bar{k}_{\ell}$ be the independent random variables (depending on the randomness of the current iteration) that describes the number of vertices in the respective cycle after the  Step~2 of an iteration. 
Due to \Cref{lem:expectedvertexDrop}, we have $\ex[\bar{k}_i]=2k_i/2^{\T}+1/2^{\T}$ for all $i\in[\ell]$. 

Let $K=\sum_{i=1}^{\ell}k_i$ and define $\mu=3K/2^{\T}$. Let $\bar{K}=\sum_{i=1}^{\ell}\bar{k}_i$ the number of remaining vertices after Step~2. By linearity of expectation, we obtain $\ex(K)\leq 2K/2^{\T}+\ell/2^{\T}\leq\mu$, where we used that the number of cycles is upper bounded by the number of vertices, i.e.,   $\ell\leq K$.

We apply Hoeffding's inequality (\Cref{lem:hoeffding}) on the $\ell$ independent random variables $\bar{k}$ that have the range $\{0,\ldots,k-1\}$ (it is deterministically guaranteed that we always remove at least one vertex from each cycle) and obtain 
\begin{align*}
    \Pr(\bar{K}\geq 6 n'/2^{\T}) & =\Pr(\bar{K}\geq 2\mu) \leq \Pr(\bar{K}-\ex(\bar{K})\geq \mu)  \\ 
    &\leq \exp\left(-\frac{2\mu^2}{\sum_{i=1}^{\ell}k_i^2}\right)
    \stackrel{(k_i\leq n^{\eps})}{\leq} \exp\left(-\frac{2\mu^2}{n^{\eps}\sum_{i=1}^{\ell}k_i}\right) \\
    &\leq  \exp\left(-\frac{2\mu^2}{ n^{\eps}\cdot K}\right) 
    \leq \exp\left(-\frac{18K}{n^{\eps}2^{2\T}}\right) \leq 1/n^4~, 
\end{align*}
where we used $K\geq n^{10\eps}$ and $\T\leq (\eps \log n)/100$ in the last step.  This proves the claim. 
\end{proof}

\subsection{Proof of Theorem~\ref{thm:main-forest}}

Let us put everything together and prove the following theorem. 

\thmmainforest*
\begin{proof}
We apply \Cref{alg:forest}. As the first step, we perform the reduction from the forest connectivity problem to the cycle connectivity problem as described in \Cref{obs: foresttocycle}.
By \Cref{lem:ShrinkLargeCycles}, after invoking \slc, we have a bound of $O(n^{\delta})$ on the longest remaining cycle.
For the rest of the proof, suppose that the total number of remaining vertices $n'$ (ignoring cycles with a single node) is at least $n^{\delta/10}$, i.e., we satisfy the requirement in \Cref{lem:whpProgress}.
Otherwise, we can collect the remaining graph onto a single machine and solve the problem locally.

Denote by $n_i$ the number of vertices after iteration $i$ and notice that $n_0 \leq 2n$ due to the reduction \Cref{obs: foresttocycle}.
Furthermore, let $\T_i \coloneqq 2 \uparrow\uparrow i$. 
Due to the design of \Cref{alg:forest}, the value of $\T$ in iteration $2i$ is more than $\min\{(\eps\log n)/100,\T_i\}$.
As long as the the cut-off at $(\eps\log n)/100$ does not happen, due to \Cref{lem:whpProgress} and a union bound, we have \whp that
\[
    n_{i + 1} \leq n_i \cdot (6/2^{\T_i})^{i+1} \leq n_0 \cdot \frac{1}{2 \uparrow\uparrow i} \leq  \frac{2n}{2 \uparrow\uparrow i}  \ .
\]
But if $\T=(\eps\log n)/100$ we obtain by \Cref{lem:whpProgress} that w.h.p.\ the number of vertices is at most $n\cdot (6/2^{\T}\leq n^{1-\eps/100)}\leq n/\log n$. 
Hence, regardless of whether the value of $\T$ is capped at $(\eps\log n)/100$ or not, after at most $O(\log^* n)$ iterations the number of vertices is at most $O(n / \log n)$. Then we can apply \sscf from \Cref{lem:ShrinkSmallCyclesFast} to finish the algorithm.

\textbf{Total Space:} 
 By the analysis of \cite{MPC-via-remote-access}, the application of \Cref{lem:ShrinkLargeCycles} requires $O(m)$ total space.
From \Cref{lem:globalQueries}, we get an upper bound on the number of queries to the \ampc hashtable, i.e., the required total space in iteration any $i$.
Let us consider two cases.
First, suppose that $i = 2j$, for some integer $j$.
Then, by the design of our algorithm and by \Cref{lem:whpProgress}, we have that $n_i \leq 6/2^{\T_j}$, where $\T_j$ corresponds to the current value of $\T$ in iteration $2j$.
Hence, by \Cref{lem:globalQueries}, we have that the required total space is $n_i \cdot \T_j = O(m)=\tspc$, where $m$ corresponds to the number of edges in the input.

Then, suppose that $i = 2j + 1$, for some integer $j$.
In this iteration, we do not increase $\T$ and hence, its value corresponds to $\T_j$.
Then, we can use the same calculations as above.

\textbf{Local space:} 
By the analysis of \cite{MPC-via-remote-access}, the application of \Cref{lem:ShrinkLargeCycles} requires $O(n^\delta)$ space per machine.
 Afterwards, all cycles are of length $O(n^\delta)$ and hence, no vertex needs to query more than $O(n^\delta)$ vertices in its component.
 Combining with the total space bound, we get the $O(n^\delta)$ bound on the required memory per machine\footnote{By using, for example, random load-balancing, we can \whp guarantee that no machine needs to collect more information than the other machines.}

\textbf{\compose:} 
Finally, we need to keep track of the mapping we create, as specified in \Cref{def:cc-shrinking}.
In a step of contraction, each vertex can keep a pointer to the vertex remaining after contraction.
These pointers are then updated after any successive contractions, requiring $O(1)$ rounds.
The pointers do not effect the asymptotic demand in runtime.

\textbf{Trading time for global memory:} We obtain that the algorithm finishes in $O(k)$ rounds if we have an additional factor of $\Omega(\log^{(k)}n)$ global memory if we initialize $\T=2 \uparrow\uparrow (c\cdot\logstar n-k)$ where $c$ is the constant in the running time of the previous algorithm. Note that the arguments about global memory and the total number of queries per iteration stay intact, but the number of iterations until we have reduced to at most $n/\log n$ vertices, reduces to at most $O(k)$.
\end{proof}

\section{General Graphs}
In this section we show our algorithm for general graphs, and prove the following. 

\thmmaingeneral*

\begin{algorithm}[t]
\caption{Algorithm for finding connected components in general graphs. $\tspc$ is the total amount of available space and $\mspc$ is the amount of available space per machine.}\label{alg:general}
\begin{algorithmic}[1]

\Function{\cc}{$G$}
    \State{Let $n= |V(G)|$, $m=|E(G)|$ and $d = \sqrt{m/n}$.}
    \If{$T / n = n^{\Omega(1)}$}\label{l:easy}
       \State Compute connected components of $G$ using algorithm of Theorem~\ref{thm:ampc-cc}.
   \EndIf
   \State{$H :=$ graph obtained by sampling each edge of $G$ independently with probability $1/d$.\label{l:sample}}

   \State{$C := \sar(H, n)$}

   \State \Return{$\compose(\sar(\contract(G, C), n), C)$}
   
\EndFunction

\Function{\sar}{$G, n$}
\State{$(G', M) := \gshrink(G, \min(2^{\sqrt{T/n}}, \sqrt{S}))$}
\State{\Return{$\compose(\cc(G'), M)$}}
\EndFunction
\end{algorithmic}
\end{algorithm}

Let us first describe the high-level ideas behind our algorithm.
Similar to the case of forests, we follow the general idea of trying to rapidly decrease the number of nodes, or equivalently as we put it in this section, increase the amount of available space per each vertex in the graph .
Once the space per vertex is large enough, we can simply use an existing algorithm using large total space.

\begin{theorem}[\cite{MPC-via-remote-access}]\label{thm:ampc-cc}
There exists an algorithm which computes connected components of an undirected graph in
$O(\log \log_{\tspc/n} n)$ \ampc rounds using total space $\tspc = \Omega(n+m)$.
\end{theorem}

Observe that when $\tspc/n = n^{\Omega(1)}$, we have $O(\log \log_{\tspc/n} n) = O(\log \log_{n^{\Omega(1)}} n) = O(\log O(1)) = O(1)$.

The starting point for increasing the amount of available space per vertex is the following lemma.

\begin{lemma}[$\gshrink$]\label{lem:shrinking-general}
Assume that the available space per machine is $\mspc$. There exists a CC-shrinking algorithm that for any parameter $1 \leq t = O(\sqrt{\mspc})$ and any $n$-vertex and $m$-edge graph $G$ outputs a graph $H$, such that $\ex(|V(H)|) = O(m / t)$, $|E(H)| = O(m)$.
The algorithm can be implemented in $O(1)$ \ampc rounds using $O(m \log t)$ space in expectation.
\end{lemma}

We use \gshrink to refer to the algorithm described in the above lemma.
The lemma with $t = \Theta(\sqrt{\mspc})$ was proven in the prior work \cite{ampc-constant} where it was used to obtain a constant-round \ampc algorithm for finding connected components using logarithmic space per vertex.
In Section~\ref{sec:shrinking-general}, we extend the algorithm and the analysis to handle the case when $1 \leq t = O(\sqrt{\mspc})$.

The challenge with applying Lemma~\ref{lem:shrinking-general} is that it does \emph{not} reduce the number of edges in the graph, and at the same time it outputs a graph, whose number of vertices depends on the number of edges in the input graph.
Hence, repeated applications of Lemma~\ref{lem:shrinking-general} do not provide stronger guarantees than a single application.
Moreover, if our goal is to use optimal space, we can only apply it with constant $t$, which does not imply any reduction in the graph size.

To address the former problem, we reduce the problem of finding connected components in a graph with average degree $r$ to two instances of a connected components problems in graphs with the same number of vertices and average degrees $O(\sqrt{r})$.
This is achieved by uniformly sampling edges, as shown in the following theorem.

\begin{theorem}[\cite{karger1995randomized}]
Let $G$ be a graph without multi-edges and let $p \in (0, 1)$.
Assume that $H$ is a random subgraph of $G$ obtained by sampling each edge of $G$ independently with probability $p$.
Then, the expected number of edges of $G$ which connect distinct connected components of $H$ is at most $n/p$.
\end{theorem}

In our algorithm we use the following simple corollary.

\begin{corollary}\label{cor:sampling}
Let $n = |V(G)|$ and $m = |E(G)|$, and let $C$ be a CC-labeling of $H$. If we set $p = \sqrt{m/n}$, then the expected number of edges in both $H$ and $\contract(G, C)$ is $O(\sqrt{mn})$.
\end{corollary}

By alternating Lemma~\ref{lem:shrinking-general} and uniform edge sampling we can show that the number of vertices decreases very quickly.
That is, roughly speaking, in one step we can increase the amount of available space per vertex of the graph from roughly $T/n$ to $2^{\sqrt{T/n}}$.
As a result, even if we start with only constant space per vertex, we can show that in $O(\logstar{n})$ rounds we get to the case when the available space per vertex is polynomially large and we can apply the algorithm of Theorem~\ref{thm:ampc-cc}.

The pseudocode of our algorithm is given as Algorithm~\ref{alg:general}.
Let us now describe the subroutines it uses.
Recall that \shrink and \compose are defined in Section~\ref{sec:preliminaries}.
Moreover, we use $\gshrink$ to refer to the CC-shrinking algorithm of Lemma~\ref{lem:shrinking-general}.

\begin{lemma}\label{lem:gcorrect}
Algorithm~\ref{alg:general} correctly computes a CC-labeling of the input graph $G$.
\end{lemma} 

\begin{proof}
The lemma follows directly from an inductive argument.
The base case holds thanks to Theorem~\ref{thm:ampc-cc}, and the inductive step follows from the fact that both \contract and \gshrink are CC-shrinking algorithms (see Observation~\ref{obs:contract} and Lemma~\ref{lem:shrinking-general}).
We will separately prove that the algorithm terminates.
\end{proof}

\subsection{Running Time}
In this section we prove the following bound on the size of the recursion in Algorithm~\ref{alg:general}.
The proof is independent of the model in which the algorithm is run.
We will discuss the aspects related to the \ampc implementation in the next section.

\begin{lemma}\label{lem:running-time}
Assume that algorithm Algorithm~\ref{alg:general} is run on an $m$-edge graph $G$ with $\tspc = \Omega(m + n \itlog{k}{n})$ total space for $k \geq 1$.
Then, the expected number of recursive \cc calls is $2^{O(k)}$.
\end{lemma}

Algorithm~\ref{alg:general} is a recursive procedure, which either returns immediately or makes exactly two recursive calls.
Let us now present the high level idea behind the proof of Lemma~\ref{lem:running-time}.
For simplicity, let us for now assume that the bounds of Corollary~\ref{cor:sampling} and Lemma~\ref{lem:shrinking-general} hold deterministically (rather than in expectation), and ignore constant factors.

Consider a call to \cc.
We will study how the amount of space per vertex ($\tspc/n$) changes in a recursive  call.
Each call to \cc makes two recursive calls to \sar, for a graph of $\sqrt{mn}$ edges (Corollary~\ref{cor:sampling}).
We assume that these calls are made one after the other, and so they both have access to the same amount of space, $T$.
Consider one of these \sar calls.
It calls \gshrink with parameter $t = \min\{2^{\sqrt{T/n}},\sqrt{S}\}$, and so, thanks to Lemma~\ref{lem:shrinking-general}, the \sar concludes by calling \cc recursively for a graph with $\sqrt{mn} / \min(2^{\sqrt{\tspc/n}}, \sqrt{\mspc})$ vertices.
Hence, in the recursive call, the available space per vertex is
\[
\tspc / (\sqrt{mn} / \min(2^{\sqrt{\tspc/n}}, \sqrt{\mspc})) \geq \min(2^{\sqrt{\tspc/n}}, \sqrt{\mspc}), 
\]
since $\tspc \geq \sqrt{mn}$, which follows from $\tspc = \Omega(n+m)$.
As a result, according to this simplified analysis, in each recursive call  of \cc we increase the amount of space per vertex from $\tspc/n$ to either $2^{\sqrt{\tspc/n}}$, or to $\sqrt{\mspc} = n^{\Omega(1)}$.
In the latter case, we ensure that the next recursive call will return immediately in line~\ref{l:easy}.
In the former case we make significant progress in increasing the amount of available space per vertex.

Note that if we define $f(x) = 2^{\sqrt{x}}$ then $f(f(x)) = \omega(2^x)$, which implies the following.

\begin{observation}
Let $f(x) := 2^{\sqrt{x}}$, and let $k \geq 2$ be an integer.
Then $f(f(\Omega(\itlog{k}{n}))) = \omega(\itlog{k-1}{n})$.
\end{observation}

As a result, after $O(k)$ recursive calls the available space per vertex increases from $\itlog{k}(n)$ to $\Omega(\log n)$, after which it increases to $\Omega(\sqrt{S})$, at which point the algorithm uses the algorithm of Theorem~\ref{thm:ampc-cc} and returns immediately.
Overall the recursion tree is a binary tree of depth $O(k)$, which leads to a running time of $2^{O(k)}$ rounds.
In the following we formalize the analysis sketched above.

We will analyze the reduction in the size of the graph after the first two levels of the recursion.
That is, consider a \cc call $c_1$, which recursively makes two calls to \cc denoted by $c_2$ and $c_3$.
These calls in turn make four \cc calls, which we denote by $c_4, c_5, c_6$ and $c_7$.
Below we show the key property of the calls $c_4, \ldots, c_7$.

\begin{lemma}\label{lem:twolevels}
Assume that \cc is called for a graph such that the available space per vertex is either (a) $\Omega(\itlog{k}{n})$ for some $k \geq 2$, or (b) $\Omega(\log n)$.
Consider the (at most) four recursive calls to \cc made after the first two levels of the recursion.
Then, with probability at least $4/5$ in all of these calls the available space per vertex is at least $\Omega(\itlog{k-1}{n})$ in case (a) or $n^{\Omega(1)}$ in case (b).
\end{lemma}

\begin{proof}
The first two recursive levels involve three \cc calls overall -- the initial one and the two calls made directly from the initial call.
As sketched above we get that the available space per vertex increases from $\Omega(\itlog{k}{n})$ to $\omega(\itlog{k-1}{n})$ or from $\Omega(\log n)$ to $n^{\Omega(1)}$ if in each of these calls the following four events happen:
$H$ has $O(\sqrt{mn})$ edges, $\contract(G, C)$ has $O(\sqrt{mn})$ edges, and in both calls to \gshrink the expected reduction in the number of edges does happen.
These are four events of the form the value of $X$ is $O(\ex[X])$.
Thanks to Markov's inequality, for a large enough hidden constant, each of these events happens with probability at least $p = 0.99$.
Since we look at three \cc calls, we overall have $12$ events which happen with probability at least $p = 0.99$.
By using union bound we have that they all hold with 
probability at least $1 - 12p = 0.88 > 4/5$.
\end{proof}

We can now prove Lemma~\ref{lem:running-time}

\begin{proof}[Proof of Lemma~\ref{lem:running-time}]
For $k \geq 1$, let us denote by $T(k)$ the number of recursive calls of the algorithm when the available space per vertex is $\Omega(\itlog{k}{n})$.
Moreover, let $T(0)$ denote the number of recursive calls when the available space per vertex is $n^{\Omega(1)}$.

We have $T(0) = 1$, and without loss of generality, we can assume that $T(k)$ is nondecreasing.
Thanks to Lemma~\ref{lem:twolevels} for $k \geq 1$ we have \[
T(k) \leq 1 + 4/5 \cdot  4\cdot T(k-1) + 1/5 \cdot 4 \cdot T(k),
\]
since with probability at least $4/5$ we increase the available space as needed and with the remaining probability (which we pessimistically upper bound by $1/5$) we, again pessimistically, assume that we make no progress in the amount of space per vertex.
By subtracting $4/5\cdot T(k)$ from both sides we get \[ 1/5 \cdot T(k) \leq 1 + 16 / 5 \cdot T(k-1),
\]
from which we conclude $T(k) = 2^{O(k)}$.
\end{proof}

\subsection{Algorithm~\ref{alg:general} in the \ampc model}

We first show that the number of recursive \cc calls in Algorithm~\ref{alg:general} directly translates to the number of rounds in the \ampc model.
We note that while each \cc call makes two recursive calls, they cannot be run in parallel, as the result of the first recursive call is needed before the second recursive call can be started.

\begin{lemma}\label{lem:gtime}
Algorithm~\ref{alg:general} can be implemented in the \ampc model, such that each \cc call, excluding its recursive calls, takes $O(1)$ \ampc rounds.
\end{lemma}

\begin{proof}
Once we have at least $n^{\Omega(1)}$ space per vertex we use the algorithm of Theorem~\ref{thm:ampc-cc}, which runs in $O(1)$ \ampc rounds.
In the remaining case we first need to sample graph $H$, which can clearly be done in $O(1)$ rounds.
In addition to that we make a constant number of calls to $\gshrink$, which uses $O(1)$ \ampc rounds (see Lemma~\ref{lem:shrinking-general}), as well as $\contract$ and $\compose$, both of which use $O(1)$ \ampc rounds as well.
\end{proof}

\begin{lemma}\label{lem:gspace}
Algorithm~\ref{alg:general} can be implemented in the \ampc model using $O(\tspc)$ total space.
\end{lemma}

\begin{proof}
There are two functions called by \cc, which may use space which is super-linear in their input sizes. We reason that these calls are still upper bounded by $O(\tspc)$. 
First, there is algorithm of Theorem~\ref{thm:ampc-cc}, which uses $O(\tspc)$ space.
Second, we call \gshrink.
Thanks to Corollary~\ref{cor:sampling} the expected number of edges passed in the argument of \shrink is $O(\sqrt{mn})$ and the second argument is upper-bounded by $2^{\sqrt{\tspc/n}}$. 
By Lemma~\ref{lem:shrinking-general} the expected space usage is $O(\sqrt{mn} \cdot \log(2^{\sqrt{\tspc/n}})) = O(\sqrt{mn} \cdot \sqrt{\tspc / n}) = O(\sqrt{\tspc m}) = O(T)$.
\end{proof}

\subsection{Proof of Lemma~\ref{lem:shrinking-general}}\label{sec:shrinking-general}
In this section we show a CC-shrinking algorithm, which is one of the building blocks of our algorithm.

The starting point is Algorithm 1 of \cite{ampc-constant}, which provides the guarantees of Lemma~\ref{lem:shrinking-general}, but uses $O(\log n)$ additional space (regardless of the choice of $t$).
Let us now describe it briefly.
It begins by transforming the input graph $G$ to a graph $G_3$ with maximum degree $3$.
This is achieved by replacing each vertex $v$ of degree $d > 3$ with a cycle of length $d$.
Each edge incident to $v$ is then connected to a different vertex of the cycle.

After that, the algorithm picks a uniformly random rank $r(v) \in [0, 1]$ for each vertex $v$ and runs BFS from each vertex, which stops as soon as one of the following conditions holds: (1) the search from $v$ explored $t$ vertices, or (2) the connected component of $v$ was fully explored, or (3) a vertex $w$ of rank lower than $v$ was reached.
Whenever the search stopped due to case (3), we add a directed \emph{super-edge} from $w$ to $v$.

One can show that the super-edges induce a forest of rooted trees, and the probability that a vertex is a root of a tree is $O(1/t)$ (Lemma 3.3 in~\cite{ampc-constant}).
The last step of the algorithm is to compute a CC-labeling $C$ of the graph defined by super-edges and return $\textsc{Contract}(G_3, C)$.
Since the number of vertices in $G_3$ is $m$, it follows directly that the expected number vertices in the resulting graph is $O(m / t)$.

We improve upon the analysis of \cite{ampc-constant} by showing that the space usage of the algorithm outlined above is $O(m \log t)$.

\begin{claim}[\cite{MPC-via-remote-access}]
The total expected space used by the BFS step is $O(m \log t)$.
\end{claim}
\begin{proof}
Let us analyze the amount of communication used by the BFS starting at some vertex $v$.
Assume that the connected component containing $v$ has size at least $t$ (otherwise the communication can only be lower). 
Observe that the BFS explores exactly $k \leq t$ vertices when the $k$th explored vertex has the smallest rank among all $k$ vertices and vertex $v$ has smallest rank among the first explored $k-1$ vertices. This happens with probability $1/(k(k-1))$.
Hence, the expected number of explored vertices is
\[
\sum_{i=2}^t \frac{i}{i (i-1)} = O(\log t).
\]
Since each vertex has constant degree we get that running BFS from all vertices requires $O(m \log t)$ expected space.
\end{proof}

\begin{claim}
There exists an algorithm which can compute connected components of the forest defined by all super-edges in $O(1)$ rounds and optimal space.
\end{claim}

\begin{proof}
We observe that the problem of finding connected components in the forest of super-edges is not a general forest connectivity problem, but rather a \emph{rooted} forest connectivity problem.
In particular, each tree of the forest has exactly one marked vertex (the root).
The forest connectivity algorithm in \cite{ampc-constant} first maps each tree to a cycle (i.e. its Euler-tour), which can be done in $O(1)$ \mpc rounds.
Then, it shrinks each cycle to ensure it has length $O(n^{\eps})$.
These transformations can be done in $O(1)$ rounds and optimal space, also see \Cref{sec:ForestConnectivity} for more details on these operations.
At this point we observe that given that we start with a collection of trees, in which each tree has a single marked vertex, we can also ensure that after the transformations we are left with a collection of cycles of length 
 $O(n^{\eps})$, in which each cycle has a single marked vertex.
 This connected components problem can be solved in a single round, as each marked vertex can simply traverse all of the cycle it belongs to and discover its entire connected component.
\end{proof}

This concludes the last step in proving \Cref{thm:maingeneral}, which follows directly by combining \Cref{lem:gcorrect,lem:running-time,lem:gtime,lem:gspace}.

\bibliographystyle{alphaurl}
\bibliography{ampc-cc-arxiv}

\end{document}